\documentclass{amsart}[10pt]

\usepackage{amsmath,amsthm,amssymb,amsfonts}
\usepackage{graphicx}
\usepackage{rotating}
\usepackage{hyperref}
\usepackage{enumerate}
\usepackage[all]{xy}
\usepackage[mathscr]{euscript}
\newtheorem{theorem}{Theorem}[section]
\newtheorem{lemma}[theorem]{Lemma}

\theoremstyle{definition}
\newtheorem{definition}[theorem]{Definition}

\newtheorem{prop}[theorem]{Proposition}
\newtheorem{cor}[theorem]{Corollary}

\theoremstyle{remark}
\newtheorem{remark}[theorem]{Remark}

\numberwithin{equation}{section}

\newcounter{stepnum}

\newcommand{\vv}{\overrightarrow}

\begin{document}

\title{Solutions to ${\rm SU}(n+1)$ Toda system generated by spherical metrics}

\author{Yiqian Shi}
\address{School of Mathematical Sciences, University of Science and Technology of China \newline \indent  Hefei, 230026, People's Republic of China}
\email{yqshi@ustc.edu.cn}

\author{Chunhui Wei$^\dagger$}
\address{School of Gifted Young, University of Science and Technology of China\newline \indent 
 Hefei, 230026, People's Republic of China}
 \email{xclw3399@mail.ustc.edu.cn}
\curraddr{School of Mathematics and Statistics, The University of Melbourne\newline \indent
Victoria, 3010, Australia}
\email{chunhuiw2@student.unimelb.edu.au}

\thanks{$^\dagger$C.W. is the corresponding author.}

\author{Bin Xu}
\address{School of Mathematical Sciences, University of Science and Technology of China \newline \indent  Hefei, 230026, People's Republic of China}
\email{bxu@ustc.edu.cn}

\date{}

\dedicatory{}

\keywords{${\rm SU}(n+1)$ Toda system, meromorphic function, rational normal map}

\begin{abstract}
Following A. B. Givental ({\it Uspekhi Mat. Nauk}, 44(3(267)):155–156, 1989), we refer to an $n$-tuple $(\omega_1,\ldots, \omega_n)$ of Kähler forms on a Riemann surface $S$ as a {\it solution to the {\rm SU}$(n+1)$ Toda system} if and only if
\[
\big(\mathrm{Ric}(\omega_1),\ldots, \mathrm{Ric}(\omega_n)\big) = (2\omega_1,\ldots, 2\omega_n)C_n,
\]
where $C_n$ is the Cartan matrix of type $A_n$. In particular, when $n=1$, this solution corresponds to a spherical metric.
Using the correspondence between solutions and totally unramified unitary curves, we show that a spherical metric $\omega$ generates a family of solutions, including $\big(i(n+1-i)\omega\bigr)_{i=1}^n$. Moreover, we characterize this family in terms of the monodromy group of the spherical metric.
As a consequence, we obtain a new solution class to the SU$(n+1)$ Toda system with cone singularities on compact Riemann surfaces, complementing the existence results of Lin-Yang-Zhong ({\it JDG}, 114(2):337-391, 2020).
\end{abstract}

\maketitle

 \section{Introduction}
 We present a natural and precise method for generating solutions to the SU$(n+1)$ Toda system on Riemann surfaces using spherical metrics (Theorems \ref{thm:rn} and \ref{thm:fam}). As a consequence, we identify a new class of solutions to the SU$(n+1)$ Toda system with cone singularities on compact Riemann surfaces (Corollary \ref{cor2.7}), which complements the results in \cite[Theorems 1.8 and 1.9]{LYZ2020}. To obtain these results, we employ the complex differential-geometric framework for solutions to the SU$(n+1)$ Toda system with cone singularities, as established in \cite{Giv1989} and \cite[Subsections 1.1 and 1.2]{MSX2024}. For further details, interested readers may refer to \cite[Section 1]{MSX2024} for the latest developments in this field.

Let \( S \) be a Riemann surface, not necessarily compact, and let \( n \) be a positive integer. An \( n \)-tuple \( \vv{\omega} = (\omega_1, \ldots, \omega_n) \) of Kähler forms is called a {\it solution to the  {\rm SU}$(n+1)$ Toda system} (\cite[Definition 1]{MSX2024}) on \( S \) if and only if
\begin{equation}
\label{equ:Toda}
\mathrm{Ric}(\vv{\omega}) = 2 \vv{\omega} C_n,
\end{equation}
where \( \mathrm{Ric}(\vv{\omega}) = \big(\mathrm{Ric}(\omega_1), \ldots, \mathrm{Ric}(\omega_n)\big) \) is the \( n \)-tuple of Ricci forms, and
\begin{equation*}
C_n =
\begin{pmatrix}
2 & -1 & 0 & \cdots & \cdots & 0 \\
-1 & 2 & -1 & 0 & \cdots & 0 \\
0 & -1 & 2 & -1 & \cdots & 0 \\
\vdots & \vdots & \vdots & \vdots & \vdots & \vdots \\
0 & \cdots & 0 & -1 & 2 & -1 \\
0 & \cdots & \cdots & 0 & -1 & 2
\end{pmatrix}_{n \times n} \label{1.2}
\end{equation*}
is the Cartan matrix of type \( A_n \). In particular, a solution \( \omega_1 \) to the SU$(2)$ Toda system coincides with a conformal spherical metric on \( S \).

In 2022, we made the simple observation that if \( \omega \) is a solution to the SU$(2)$ Toda system, then \( \left(i(n+1-i)\omega\right)_{i=1}^n \) solves the SU$(n+1)$ Toda system on \( S \). In this paper, we will develop this strategy in detail using the \textit{basic correspondence between solutions to the {\rm SU}$(n+1)$ Toda system on \( S \) and the totally unramified unitary curves from \( S \) to the complex projective space \( \mathbb{P}^n \) of dimension \( n \)}. The definition of a totally unramified unitary curve and the proof of this correspondence can be found in \cite[Subsection 1.2 and Section 2]{MSX2024}. Simply put, a totally unramified unitary curve \( f:S \to \mathbb{P}^n \) is a multi-valued holomorphic map whose monodromy group resides within \( \mathrm{PSU}(n+1) \) and any local germs are totally unramified. We also refer to a unitary curve corresponding to a solution as an \textit{associated curve of the solution}. Any two associated curves of a solution differ by a rigid motion of \( \mathbb{P}^n \) endowed with the Fubini-Study metric \( \omega_{\rm FS} \) (\cite[(4.12)]{Griffiths:1974}). In particular, an associated curve of the solution \( \omega \) to the SU$(2)$ Toda system coincides with the developing map of the spherical metric \( \omega \) on \( S \) (\cite[Section 2]{CWWX:2015}). First, we characterize unitary curves \( S \to \mathbb{P}^n \) associated with the solution \( \left(i(n+1-i)\omega\right)_{i=1}^n \) in terms of a unitary curve \( S \to \mathbb{P}^1 \) associated with \( \omega \).

\begin{theorem}\label{thm:rn}
Let \( \omega \) be a solution to the {\rm SU}$(2)$ Toda system on \( S \), and let \( v: S \to \mathbb{P}^1 \) be a curve associated with \( \omega \). Let \( r_n: \mathbb{P}^1 \to \mathbb{P}^n \) be the rational normal map defined by
$$
r_n: [z_0, z_1] \mapsto \left[ \sqrt{\frac{1}{n!}} z_0^n, \sqrt{\frac{1}{(n-1)!1!}} z_0^{n-1}z_1, \cdots, \sqrt{\frac{1}{n!}} z_1^n \right].
$$
Then \( \big(i(n+1-i)\omega\big)_{i=1}^n \) solves the {\rm SU}$(n+1)$ Toda system on \( S \). Moreover, the set
$$
\left\{ U \circ r_n \circ v : S \to \mathbb{P}^n \mid U \in \mathrm{PSU}(n+1) \right\}
$$
consists of all the associated curves of this solution.
\end{theorem}

Given a basis of \( \mathbb{C}^{n+1} \) endowed with the standard Hermitian inner product \( \langle \,, \, \rangle \), the Gram-Schmidt procedure provides a new orthonormal basis of \( \big( \mathbb{C}^{n+1}, \langle \,, \, \rangle \big) \). Then, we obtain the Iwasawa decomposition of \( \mathrm{SL}(n+1, \mathbb{C}) \) in the form
$$
\mathrm{SL}(n+1, \mathbb{C}) =   \mathrm{SU}(n+1)\,\Delta_{n+1},
$$
where \( \mathrm{SU}(n+1) \) is the group of special unitary transformations and \( \Delta_{n+1} \) is the group of linear transformations by left multiplication of upper triangular matrices with positive diagonal entries in \( \mathrm{SL}(n+1, \mathbb{C}) \). Hence, an automorphism \( \varphi \in \mathrm{PSL}(n+1, \mathbb{C}) \) of \( \mathbb{P}^n \) has the decomposition \( \varphi = U \circ \delta\), where \( U \in \mathrm{PSU}(n+1) \) and \( \delta \in \Delta_{n+1} \). Based on this and Theorem \ref{thm:rn}, we introduce the following definition:

\begin{definition}
We call a solution \( \vv{\omega} \) to the SU$(n+1)$ Toda system \textit{reduced} if and only if it is generated by another solution \( \omega_1 \) to the SU$(2)$ Toda system, i.e., a conformal spherical metric, on \( S \) in the following sense: there exists a linear transformation \( \delta \in \Delta_{n+1} \), an associated curve \( f: S \to \mathbb{P}^n \) of \( \vv{\omega} \), and an associated curve \( v: S \to \mathbb{P}^1 \) of \( \omega_1 \) such that
\begin{equation}
\label{equ:red}
f = \delta \circ r_n \circ v.
\end{equation}
Notably, the curve \( f \) should have monodromy in \( \mathrm{PSU}(n+1) \), which imposes a constraint on the variety of such \( \delta \)'s (Theorem \ref{thm:fam}).
\end{definition}

Given a solution $\omega_1$ to the SU$(2)$ Toda system, we can characterize all
the reduced solutions to the SU$(n+1)$ Toda system generated by it in the following theorem:

\begin{theorem}\label{thm:fam} We use the notions in Theorem \ref{thm:rn}.
Let 
\[M_v=\left\{\varphi\in{\rm PSL}(n+1,\mathbb{C})|\varphi\circ r_n\circ v:S\to\mathbb{P}^n\ \text{is a unitary curve}\right\}.\]
$M_v$ can be decomposed into $M_v=\mathrm{PSU}(n+1)\Delta_v$, where $\Delta_v\subset\Delta_{n+1}$. It is determined by the closure $\overline{G_v}$  in ${\rm PSU}(2)$ of the monodromy group $G_v$ of $v$. Consider the classification of closed subgroups of $\mathrm{SU}(2)$ {\rm (\cite[Chapter 1]{pietro_giuseppe_frè_2018})}{\rm :}
\[
\begin{aligned}
    {\rm O}(2) &= \left\langle {\rm U}(1),\begin{pmatrix}0&-1\\1&0\end{pmatrix} \right\rangle, \quad {\rm U}(1),\\
    C_k &= \left\langle \begin{pmatrix}e^{2\pi \sqrt{-1}/k}&0\\0&e^{-2\pi \sqrt{-1}/k}\end{pmatrix} \right\rangle, k\in\mathbb{Z}_{>0},\\
    D_k &= \left\langle C_{2k},\begin{pmatrix}0&-\sqrt{-1}\\\sqrt{-1}&0\end{pmatrix} \right\rangle, k\in\mathbb{Z}_{>0},\\
    E_6 &= \left\langle \begin{pmatrix}
        \frac{1+\sqrt{-1}}{2}&\frac{1+\sqrt{-1}}{2}\\
        \frac{\sqrt{-1}-1}{2}&\frac{1-\sqrt{-1}}{2}
        \end{pmatrix},\begin{pmatrix}
        0&-\sqrt{-1}\\
        -\sqrt{-1}&0
        \end{pmatrix} \right\rangle, \\
    E_7 &= \left\langle \begin{pmatrix}
        \frac{1+\sqrt{-1}}{2}&\frac{1+\sqrt{-1}}{2}\\
        \frac{\sqrt{-1}-1}{2}&\frac{1-\sqrt{-1}}{2}
        \end{pmatrix},\begin{pmatrix}
        0&\frac{\sqrt{2}(1+\sqrt{-1})}{2}\\
        \frac{\sqrt{2}(\sqrt{-1}-1)}{2}&0
        \end{pmatrix} \right\rangle, \\
    E_8 &= \left\langle \begin{pmatrix}
        \frac{1}{2}&-\frac{\sqrt{5}-1}{4}+\frac{\sqrt{5}+1}{4}\sqrt{-1}\\
        \frac{\sqrt{5}-1}{4}+\frac{\sqrt{5}+1}{4}\sqrt{-1}&\frac{1}{2}
        \end{pmatrix},\begin{pmatrix}
        0&-\sqrt{-1}\\
        -\sqrt{-1}&0
        \end{pmatrix} \right\rangle.
\end{aligned}
\]
Let $p:{\rm SU}(2)\to{\rm PSU}(2)$ be the quotient map. Then there hold the following statements{\rm :}
\begin{enumerate}
    \item When $\overline{G_v}=\mathrm{PSU}(2)$, $\Delta_v =\{I_{n+1}\}$;
    \item When $\overline{G_v}=\mathrm{PU}(1)$, $\Delta_v = \{\mathrm{diag}(a_0, \cdots, a_{n}) \in \Delta_{n+1}\}$ with $\dim_{\mathbb{R}} \Delta_v=n$;
    \item When $\overline{G_v}=\mathrm{PO}(2)$, $\Delta_v = \{\mathrm{diag}(a_0, \cdots, a_{n})\in \Delta_{n+1} | a_i = a_{n-i}\}$ with $\dim_{\mathbb{R}}\Delta_v=\lfloor n/2\rfloor$;
    \item When $\overline{G_v}=p(C_k)$, $\Delta_v= \{(a_{ij})_{0\leq i,j\leq n} \in \Delta_{n+1} | a_{i,j} = 0 \text{ if } \frac{k}{\gcd(k,2)} \nmid (i-j) \}$
    with
    \[
    \dim_{\mathbb{R}}\Delta_v=-\frac{k}{\gcd(k,2)}\lfloor\frac{n\gcd(k,2)}{k}\rfloor^2+(2n+2-\frac{k}{\gcd(k,2)})\lfloor\frac{n\gcd(k,2)}{k}\rfloor+n;
    \]
    \item When $\overline{G_v}=p(D_k)$,
    $$\Delta_v = \left\{(a_{i,j})_{0\leq i,j\leq n} \in \Delta_{n+1} \middle| \begin{aligned}
\sum_{l=0}^n \bar{a}_{l,i}a_{l,j} &= 0 \text{ if } k \nmid i-j\\
    \sum_{l=0}^n \bar{a}_{l,i}a_{l,j} &= (-\sqrt{-1})^{i-j}\sum_{l=0}^n \bar{a}_{l,n-i}a_{l,n-j}
\end{aligned}\right\}$$
    with $\dim_{\mathbb{R}} \Delta_v =-\frac{k}{2}\lfloor\frac{n}{k}\rfloor^2+(n+1-\frac{k}{2})\lfloor\frac{n}{k}\rfloor+\lfloor\frac{n}{2}\rfloor$;
    \item When $\overline{G_v}=p(E_6)$,
    \[
    \dim_{\mathbb{R}}\Delta_v=\begin{cases}
        \dim_{\mathbb{R}}\Delta_v=\frac{n^2}{12}+\frac{2c^2}{3}+\frac{n}{6}-\frac{11}{12}\quad\text{if $n$ is odd}\\
        \dim_{\mathbb{R}}\Delta_v=\frac{n^2}{12}+\frac{2c^2}{3}+\frac{n}{6}-\frac{2}{3}\quad\text{if $n$ is even}
    \end{cases},
    \]
    where $c=\cos{\frac{n\pi}{3}}+\frac{\sqrt{3}}{3}\sin{\frac{n\pi}{3}}$;
    \item When $\overline{G_v}=p(E_7)$,
    \[
    \dim_{\mathbb{R}}\Delta_v=\begin{cases}
        \frac{n^2}{24}+\frac{n}{12}+\frac{c_1^2}{3}+\frac{c_2^2}{4}-\frac{23}{24}\quad\text{if $n$ is odd}\\
        \frac{n^2}{24}+\frac{n}{12}+\frac{c_1^2}{3}+\frac{c_2^2}{4}-\frac{7}{12}\quad\text{if $n$ is even}
    \end{cases},
    \]
    where $c_1=\cos{\frac{n\pi}{3}}+\frac{\sqrt{3}}{3}\sin{\frac{n\pi}{3}}$ and $c_2=\cos{\frac{n\pi}{4}}+\sin{\frac{n\pi}{4}}$;
    \item When $\overline{G_v}=p(E_8)$,
    \[
    \dim_{\mathbb{R}}\Delta_v=\begin{cases}
        \frac{n^2}{60}+\frac{n}{30}+ \frac{c_1^2}{3}+ \frac{c_2^2}{5}+\frac{c_3^2}{5}-\frac{59}{60}\quad\text{if $n$ is odd}\\
        \frac{n^2}{60}+\frac{n}{30}+ \frac{c_1^2}{3}+ \frac{c_2^2}{5}+\frac{c_3^2}{5}-\frac{13}{15}\quad\text{if $n$ is even}
    \end{cases},
    \]
    where $c_1=\cos{\frac{n\pi}{3}}+\frac{\sqrt{3}}{3}\sin{\frac{n\pi}{3}}, c_2=\sqrt{1+\frac{2}{\sqrt{5}}} \sin{\frac{n\pi}{5}}+\cos{\frac{n\pi}{5}}$ and $c_3=\sqrt{1-\frac{2}{\sqrt{5}}} \sin{\frac{2n\pi}{5}}+\cos{\frac{2n\pi}{5}}$.
\end{enumerate}
\end{theorem}
\begin{remark}
    Notice that all the cases are possible: for any closed subgroup of $\mathrm{PSU}(2)$, there exists a multi-valued meromorphic function such that the closure of its monodromy group is the given subgroup. The dimensions of cases (6)–(8) arise from norms of characters of finite-dimensional representations of finite groups; hence, they are naturally integers, although their forms appear complicated.

\end{remark}

As an application of Theorems \ref{thm:rn} and \ref{thm:fam}, we identify a novel class of solvable SU$(n+1)$ Toda systems with cone singularities on compact Riemann surfaces as follows:

\begin{cor}\label{cor2.7}
We adopt the notions introduced in \cite[Subsection 1.1]{MSX2024}. Suppose that there exists a cone spherical metric that represents the real divisor \( D = \sum_{j=1}^n \gamma_j [P_j] \), where \( 0 \neq \gamma_j > -1 \) for all \( 1 \leq j \leq n \), on a compact Riemann surface \( X \). Then, for each positive integer \( n > 1 \), the SU$(n+1)$ Toda system on \( X \) with cone singularities
$$
\underset{n \, \text{divisors}}{\underbrace{\big(D, D, \ldots, D\big)}}
$$
has a family of reduced solutions, including \( \big(i(n+1-i)\omega\big)_{i=1}^n \) and is characterized in Theorem \ref{thm:fam}.
\end{cor}


We organize the remainder of this paper as follows. In Section 2, we prove Theorem \ref{thm:rn} using the infinitesimal Pl\"ucker formula (\cite[p. 269]{griffiths2014principles}) and the symmetric product representation of SU(2) \cite{fulton_harris_2004}. We classify all the reduced solutions generated by a spherical metric in terms of its monodromy in \( \mathrm{PSU}(2) \), and then prove Theorem \ref{thm:fam} in Section 3 by using the characters of some symmetric product representations of $E_6, E_7$ and $E_8$. In the final section, we present new solvable SU$(n+1)$ Toda systems with cone singularities on both the Riemann sphere and compact Riemann surfaces of positive genus.

\section{Existence of reduced solutions}
In this section, we prove Theorem \ref{thm:rn}. In particular, we first perform some preliminary calculations on the Wronskian of curves in $\mathbb{C}^{n+1}$, followed by proving the theorem in a local coordinate system. Finally, we apply representation theory and complete the proof on the entire Riemann surface.

\subsection{Computation of Wronskian}  
Assume that $U$ is a domain of $\mathbb{C}$.
\begin{lemma}\label{lm:wr2}
Let $f=(f_0,\cdots,f_n):U\to\mathbb{C}^{n+1}$ be a holomorphic curve and $v:U\to\mathbb{C}$ be a meromorphic function. Then the curve $v\cdot f:=(vf_0,\cdots,vf_n)$ satisfies
\[
\Lambda_n(v\cdot f) = v^{n+1} \Lambda_n(f).
\]
\end{lemma}
\begin{proof}
Omitted.
\end{proof}
\begin{lemma}\label{lm:wr3}
Let $v:U\to\mathbb{C}$ be a non-degenerate meromorphic function and $f = \left( 1, \frac{1}{1!}v, \cdots, \frac{1}{n!}v^n \right): U \to \mathbb{C}^{n+1}$. Then
\[
\Lambda_n(f) = (v^\prime)^{\frac{n(n+1)}{2}}.
\]
\end{lemma}
\begin{proof}
We prove it by induction.
\begin{enumerate}
\item Case $n=1$ is easy.\\

\item Suppose that $n\geq 2$ and for all $1\leq k\leq n-1$, we have
$$\Lambda_k\left( 1,\frac{1}{1!}v,\cdots,\frac{1}{k!}v^k\right)=(v^\prime)^{\frac{k(k+1)}{2}}.$$
Then
\begin{equation}
\begin{aligned}
\Lambda_n(f)&=\Lambda_n\left( 1,\frac{1}{1!}v,\cdots,\frac{1}{n!}v^n\right)\\
			&=\Lambda_{n-1}\left(\frac{1}{1!}v^\prime,\cdots,\frac{1}{n!}nv^{n-1}v^\prime\right)\\
			&=(v^\prime)^{n+1}\Lambda_{n-1}\left(\frac{1}{1!},\cdots,\frac{1}{(n-1)!}v^{n-1}\right)(\text{by Lemma }\ref{lm:wr2})\\
			&=(v^\prime)^{n+1}(v^\prime)^{\frac{n(n-1)}{2}}\\
			&=(v^\prime)^{\frac{n(n+1)}{2}}
\end{aligned}
\nonumber
\end{equation}
\end{enumerate}
\end{proof}
\begin{lemma}\label{lm:wr1}
Let $v_0, v_1: U \to \mathbb{C}$ be holomorphic functions such that 
\[
v_0(z)v'_1(z) - v_1(z)v'_0(z) \equiv 1 \quad \text{on } U.
\]
Then the canonical lifting
\[
f = \left( \sqrt{\frac{1}{n!}} v_0^n, \sqrt{\frac{1}{(n-1)!1!}} v_0^{n-1}v_1, \cdots, \sqrt{\frac{1}{n!}} v_1^n \right): U \to \mathbb{C}^{n+1}
\]
has Wronskian $\equiv 1$.
\end{lemma}
\begin{proof}
Let $v=v_1/v_0$. Then we have $v^\prime=\frac{1}{v_0^2}$ and
\begin{equation}
\begin{aligned}
\Lambda_n(f)&=\Lambda_n\left(\sqrt{\frac{1}{n!}}\,v_0^n:\sqrt{\frac{1}{(n-1)!1!}}\, v_0^{n-1}v_1:
\cdots:\sqrt{\frac{1}{n!}}\, v_1^n\right)\\
			&=\Lambda_n\left(\sqrt{\frac{1}{n!}}\,(v^\prime)^{-\frac{n}{2}}:\sqrt{\frac{1}{(n-1)!1!}}\, (v^\prime)^{-\frac{n}{2}}v:
\cdots:\sqrt{\frac{1}{n!}}\, (v^\prime)^{-\frac{n}{2}}v^n\right)\\
			&=(v^\prime)^{-\frac{n(n+1)}{2}}\Lambda_n\left(\sqrt{\frac{1}{n!}}\,:\sqrt{\frac{1}{(n-1)!1!}}\, v:
\cdots:\sqrt{\frac{1}{n!}}\, v^n\right)(\text{by Lemma }\ref{lm:wr2})\\
			&=(v^\prime)^{-\frac{n(n+1)}{2}}(v^\prime)^{\frac{n(n+1)}{2}}(\text{by Lemma }\ref{lm:wr3})\\
			&=1.
\end{aligned}
\nonumber
\end{equation}
\end{proof}
\subsection{Reduced solutions on a chart}  
Let $\{U, z\}$ be a complex coordinate chart of $S$. Assume that $\vv{\omega} = (\omega_1 = \frac{\sqrt{-1}}{2} e^{u_1} dz \wedge d\bar{z}, \cdots, \omega_n = \frac{\sqrt{-1}}{2} e^{u_n} dz \wedge d\bar{z})$ in $U$. Then the {\rm SU}$(n+1)$ Toda system (\ref{equ:Toda}) takes the following form:
\begin{equation}\label{localToda}
    \left(\frac{\partial^2 u_1}{\partial z \partial \bar{z}}, \cdots, \frac{\partial^2 u_n}{\partial z \partial \bar{z}}\right) = -(e^{u_1}, \cdots, e^{u_n}) C_n.
\end{equation}
Thus, we also call $(u_1, \cdots, u_n)$ a solution to the {\rm SU}$(n+1)$ Toda system on $U$. We now prove the existence of reduced solutions on $U$.
\begin{lemma}\label{lem 2.2}
Let \( \omega=\frac{\sqrt{-1}}{2}e^udz\wedge d\bar{z} \) be a solution to the {\rm SU}$(2)$ Toda system on \( U \), and let \( v: U \to \mathbb{P}^1 \) be a curve associated with \( \omega \). Let \( r_n: \mathbb{P}^1 \to \mathbb{P}^n \) be the rational normal map defined by
$$
r_n: [z_0, z_1] \mapsto \left[ \sqrt{\frac{1}{n!}} z_0^n, \sqrt{\frac{1}{(n-1)!1!}} z_0^{n-1}z_1, \cdots, \sqrt{\frac{1}{n!}} z_1^n \right].
$$
Then \( \big(i(n+1-i)\omega=\frac{\sqrt{-1}}{2}e^{u+\ln(i(n+1-i))}dz\wedge d\bar{z}\big)_{i=1}^n\) solves the {\rm SU}$(n+1)$ Toda system on \( U \). Moreover, the set
$$
\left\{ U \circ r_n \circ v : U \to \mathbb{P}^n \mid U \in \mathrm{PSU}(n+1) \right\}
$$
consists of all the associated curves of this solution.
\end{lemma}
\begin{proof}
A direct computation shows that $\big(u+\ln(i(n+1-i))\big)_{i=1}^n$ solves (\ref{localToda}).
Denote by $v=[v_0:\, v_1]$ the curve $v:U\to \mathbb{P}^1$ associated to $u$ such that $v_0(z)v'_1(z)-v_1(z)v'_0(z)\equiv 1$ on $U$. By Lemma \ref{lm:wr1}, the canonical lifting
$$\hat{f}=\left(\sqrt{\frac{1}{n!}}\,v_0^n:\sqrt{\frac{1}{(n-1)!1!}}\, v_0^{n-1}v_1:
\cdots:\sqrt{\frac{1}{n!}}\, v_1^n\right):U\to \mathbb{C}^{n+1}$$
of the curve
$$f=r_n\circ v=\left[\sqrt{\frac{1}{n!}}\,v_0^n:\sqrt{\frac{1}{(n-1)!1!}}\, v_0^{n-1}v_1:
\cdots:\sqrt{\frac{1}{n!}}\, v_1^n\right]:U\to \mathbb{P}^n$$
has Wronskian $\equiv 1$  i.e. $\hat{f}\wedge \hat{f}'\wedge\cdots\wedge \hat{f}^{(n)}\equiv e_0\wedge\cdots\wedge e_n$ on $U$(It also means that $f$ is totally unramified).\par
It suffices to check that $u+\ln n$ equals
the first component $u_1$ of
solution $(u_1,\cdots, u_n)$ of (\ref{localToda}) from the lifting $\hat f$ of the curve $f$. We have
\begin{equation}
\begin{aligned}
 u_1&=\log\left(\frac{\| \Lambda_1(\hat{f})\|^2}{\|\hat{f}\|^4}\right)\\
	&=\log\left(\frac{\partial^2}{\partial z\partial \bar{z}}\log\| \hat{f}\|^2\right)\\
	&=\log\left(\frac{\partial^2}{\partial z\partial \bar{z}}\log\Big(\frac{1}{n!}(|v_0|^2+|v_1|^2)^n\Big)\right)\\
	&=\log\left(\frac{\partial^2}{\partial z\partial \bar{z}}\log(|v_0|^2+|v_1|^2)\right)+\ln n\\
	&=u+\ln n,
\end{aligned}
\end{equation}
where we use the infinitesimal Pl\" ucker formula (\cite[p.269]{griffiths2014principles}) in the second equality.
\end{proof}
\subsection{Reduced solutions on Riemann surface} 
Then we achieve the global result considering the monodromy. Firstly, let us recall some facts about the symmetric product space.
\begin{definition}\cite[p.50]{chern_chen_kai-shue-lam_1999}
Let \( V \) be a vector space over \( \mathbb{C} \). The \( k \)-th \textbf{symmetric product} of \( V \), denoted \( \mathrm{Sym}^k(V) \), is the subspace of the \( k \)-fold tensor product space \( V^{\otimes k} \) consisting of all tensors that are invariant under the action of the symmetric group \( S_k \). Formally,
$
\mathrm{Sym}^k(V) = \left\{ T \in V^{\otimes k} \mid \sigma(T) = T, \, \forall \, \sigma \in S_k \right\},
$
where \( \sigma \) acts on $V^{\otimes k}$ by permuting arguments
$\sigma(v_1\otimes v_2\otimes \cdots\otimes v_k) = v_{\sigma(1)}\otimes v_{\sigma(2)}\otimes \cdots\otimes v_{\sigma(k)}$
for any $v_1,\cdots,v_k\in V$.
\end{definition}

\begin{definition}\cite[Definition 2.5]{chern_chen_kai-shue-lam_1999}
The \textbf{symmetrization operator} is a map that projects any tensor \( T \in V^{\otimes k} \) onto its symmetric part. It is defined as
\[
S^k(T) = \frac{1}{k!} \sum_{\sigma \in S_k} \sigma(T).
\]
\end{definition}

\begin{prop}\cite[Theorem 2.2]{chern_chen_kai-shue-lam_1999}\quad
\begin{enumerate}
    \item If \( \{e_1, e_2, \ldots, e_n\} \) is a basis of \( V \), then a basis of \( \mathrm{Sym}^k(V) \) consists of
    \[
    \left\{ S^k(e_{i_1} \otimes e_{i_2} \otimes \cdots \otimes e_{i_k}) \mid i_1 \leq i_2 \leq \cdots \leq i_k \right\}.
    \]

    \item The dimension of \( \mathrm{Sym}^k(V) \) is $\binom{n + k - 1}{k}$ with $n=\dim\, V$.
    
\end{enumerate}
\end{prop}
\begin{definition}\cite{fulton_harris_2004}
Let \( G \) be a group, and let \( V \) be a finite-dimensional vector space over \( \mathbb{C} \), equipped with a representation of \( G \):
\[
\rho: G \to \mathrm{GL}(V),
\]
where \( \mathrm{GL}(V) \) is the general linear group of \( V \). The \( k \)-th \textbf{symmetric product representation} of \( G \), denoted \( \mathrm{Sym}^k(V) \), is defined as the natural induced representation of \( G \) on the \( k \)-th symmetric product space \( \mathrm{Sym}^k(V) \), which is a subspace of \( V^{\otimes k} \).
The action of \( G \) on \( \mathrm{Sym}^k(V) \) is given by:
\[
g \cdot S^k(v_1 \otimes v_2 \otimes \cdots \otimes v_k) = S^k((g \cdot v_1) \otimes (g \cdot v_2) \otimes \cdots \otimes (g \cdot v_k)),
\]
for all \( g \in G \), \( v_1, v_2, \ldots, v_k \in V \).
\end{definition}
\begin{definition}\label{innerproduct}
Let $V$ be a vector space over $\mathbb{C}$ equipped with a Hermitian inner product $\langle \cdot, \cdot \rangle_V$. Then, for tensors \( v_1 \otimes v_2 \otimes \cdots \otimes v_n \) and \( w_1 \otimes w_2 \otimes \cdots \otimes w_n \) in \( V^{\otimes n} \), the Hermitian inner product on the tensor product space $V^{\otimes n}$ is defined as
\[
\langle v_1 \otimes v_2 \otimes \cdots \otimes v_n, w_1 \otimes w_2 \otimes \cdots \otimes w_n \rangle_{V^{\otimes n}} = \prod_{i=1}^n \langle v_i, w_i \rangle_V, 
\]
which induces a Hermitian inner product on the subspace $\mathrm{Sym}^n(V)$ of $V^{\otimes n}$.
\end{definition}
\begin{lemma}
    If $\{e_1,\cdots,e_n\}$ is an orthonormal basis of $V$, then an orthonormal basis of \( \mathrm{Sym}^k(V) \) consists of $\left\{ \sqrt{\frac{i_1!\cdots i_k!}{k!}}S^k(e_1^{\otimes i_1}\otimes\cdots\otimes e_n^{\otimes i_n}) \right\}$, where $0\leq i_1,\cdots,i_n\leq k$ and $i_1 + i_2 + \cdots + i_n=k$. Then we obtain the corresponding homogeneous coordinates on both $\mathbb{P}(V)$ and $\mathbb{P}\left(\mathrm{Sym}^k(V)\right)$.
\end{lemma}
\begin{proof}
    Of course $\{S^k(e_1^{\otimes i_1}\otimes\cdots\otimes e_n^{\otimes i_n})\}$ forms a basis of \( \mathrm{Sym}^k(V) \) and $<S^k(e_1^{\otimes i_1}\otimes\cdots\otimes e_n^{\otimes i_n}),S^k(e_1^{\otimes j_1}\otimes\cdots\otimes e_n^{\otimes j_n})>\neq 0$ if and only if $i_1=j_1,\cdots,i_n=j_n$. In addition, $<S^k(e_1^{\otimes i_1}\otimes\cdots\otimes e_n^{\otimes i_n}),S^k(e_1^{\otimes i_1}\otimes\cdots\otimes e_n^{\otimes i_n})>=\binom{k}{i_1,\cdots,i_n}=\frac{k!}{i_1!\cdots i_k!}$. Thus, $\left\{ \sqrt{\frac{i_1!\cdots i_k!}{k!}}S^k(e_1^{\otimes i_1}\otimes\cdots\otimes e_n^{\otimes i_n}) \right\}$ forms an orthonormal basis.
\end{proof}
\begin{definition}
    For a projective space \( \mathbb{P}(V) \), the Fubini-Study metric can be described as follows:
\begin{enumerate}
    \item In homogeneous coordinates \( [u] \in \mathbb{P}(V) \), the Fubini-Study distance between two points \( [u], [v] \in \mathbb{P}(V) \) is given by $d_{\text{FS}}([u], [v]) = \arccos \left( \frac{|\langle u, v \rangle|^2}{\langle u, u \rangle \langle v, v \rangle} \right)$, where \( \langle u, v \rangle \) is the Hermitian inner product on \( V \).
    \item The associated Kähler form \( \omega_{\text{FS}} \) is given by $\omega_{\text{FS}} = \sqrt{-1} \, \partial \bar{\partial} \log \langle u, u \rangle$,
    where \( \langle u, u \rangle \) is the norm square of the vector \( u \in V \).
\end{enumerate}
\end{definition}

\begin{lemma}\label{ind-embed}
    Let $V$ be a $\mathbb{C}$-Hermitian space of dimension $2$. The rational normal map $r_n: \mathbb{P}(V) \to \mathbb{P}(\mathrm{Sym}^n(V)),[u] \mapsto \left[u^{\otimes n}\right]$ induces a Lie group monomorphism $\sigma: \mathrm{PSU}(V) \to \mathrm{PSU}(\mathrm{Sym}^n(V))$
    such that $r_n \circ U = \sigma(U) \circ r_n $ for any $U \in \mathrm{PSU}(V).$
\end{lemma}
\begin{proof}
    Since
    \[
    \begin{aligned}
    r_n^*\omega_{\mathrm{FS}, \mathbb{P}(\mathrm{Sym}^n(V))}&=\sqrt{-1}\partial\overline{\partial}  \log \langle u^{\otimes n}, u^{\otimes n} \rangle\\
    &=\sqrt{-1}\partial\overline{\partial}  \log \langle u, u \rangle^n\\
    &=n \omega_{\mathrm{FS}, \mathbb{P}(V)}    
    \end{aligned},
    \]
    we have $(r_n\circ U)^*g_{\mathrm{FS}, \mathbb{P}(\mathrm{Sym}^n(V))}=nU^*g_{\mathrm{FS}, \mathbb{P}(V)}=ng_{\mathrm{FS}, \mathbb{P}(V)}=r_n^*g_{\mathrm{FS}, \mathbb{P}(\mathrm{Sym}^n(V))}$ for any $U\in\mathrm{PSU}(V)$. By the rigidity theorem \cite[(4.12)]{griffiths2014principles}, there exists a unique $U^\prime\in\mathrm{PSU}(\mathrm{Sym}^n(V))$ such that $r_n \circ U = U^\prime \circ r_n$. Defining $\sigma(U)=U^\prime$, we are done.

    In addition, let \( e_0, e_1 \) be an orthonormal basis of \( V \), and let \( [u] = [z_0 e_0 + z_1 e_1] \). Then, we have $[u^{\otimes n}] = \left[ \sum_{i=0}^n z_0^{n-i} z_1^i S^n(e_0^{\otimes n-i} \otimes e_1^{\otimes i}) \right]$.
Notice that the orthonormal basis for the symmetric powers is given by $\left\{ \sqrt{\frac{i!(n-i)!}{n!}} S^n(e_0^{\otimes n-i} \otimes e_1^{\otimes i}) \right\}$. Thus, this map corresponds to the rational normal map $$[z_0 : z_1] \mapsto \left[ \sqrt{\frac{1}{n!}} z_0^n : \sqrt{\frac{1}{1!(n-1)!}} z_0^{n-1} z_1 : \cdots : \sqrt{\frac{1}{n!}} z_1^n \right].$$
\end{proof}
\begin{proof}[Proof of Theorem \ref{thm:rn}] Since $\omega$ solves the SU(2) Toda system on $S$, it
is straightforward for us to verify that $\big(i(n+1-i)\omega\big)_{i=1}^n$ is a solution to the SU$(n+1)$ Toda system on $S$. 

Consider a chart $U$ with a branch $v_0$ of $v$ on $U$. Then, by Lemma \ref{lem 2.2}, $r_n \circ v_0$ is an associated curve of this solution restricted to $U$. Furthermore, $r_n \circ v_0$ is a branch of $r_n \circ v$ on $U$. Therefore, we need to prove that $r_n \circ v$ is a unitary curve.

For $z \in S$, since the monodromy of $v$ belongs to the group ${\rm PSU}(2)$, there exists a special unitary representation $\rho: \pi_1(S,z) \to {\rm PSU}(V,H)$, where $V$ is the natural representation space ${\mathbb{C}}^2$ of $\mathrm{PSU}(2)$, and $H$ is the Hermitian inner product on $V$ (with $v$ being viewed as a map $v: S \to \mathbb{P}(V)$). There is a symmetric product representation $\rho^\prime = \sigma \circ \rho: \pi_1(S,z) \to {\rm PSU}(\mathrm{Sym}^n(V),H^\prime)$, where $\sigma: {\rm PSU}(V,H) \to {\rm PSU}(\mathrm{Sym}^n(V),H^\prime)$ is the embedding induced by $r_n$ (Lemma \ref{ind-embed}), and $H^\prime$ is the Hermitian inner product on $\mathrm{Sym}^n(V)$ induced from $H$ (Definition \ref{innerproduct}). 

Assume that the monodromy representation of $r_n\circ v$ is $\varrho:\pi_1(S,z) \to {\rm PSL}(\mathrm{Sym}^n(V))$. For $\gamma\in\pi_1(S,z)$ and a branch $v_0$ of $v$ near $z$, if we extend $v_0$ analytically along $\gamma$, we get $\rho(\gamma)\circ v_0$. Thus, for a branch $r_n\circ v_0$ of $r_n\circ v$, if we extend $r_n\circ v_0$ analytically along $\gamma$, we get both $\varrho(\gamma)\circ r_n\circ v_0$ and $r_n\circ\rho(\gamma)\circ v_0$, which means $\varrho(\gamma)\circ r_n=r_n\circ \rho(\gamma)=\rho^\prime(\gamma)\circ r_n$. Since $r_n$ is non-degenerate, $\varrho=\rho^\prime$ is a unitary representation. So $r_n\circ v$ is a unitary curve.
\end{proof}

\section{Classification of reduced solutions}
In this section, we prove Theorem \ref{thm:fam}. Firstly, we describe \( M_v \) by the closure of the monodromy group. Then, we achieve the classification from the classification of the closure. Finally, based on the classification,  we compute the real dimension of \( M_v \). Moreover, we also use the characters of the symmetric product representations of the natural two-dimensional representations 
of $E_6, E_7$ and $E_8$.
Denote by $C(S)$ the centralizer of a subset $S\subset\mathrm{PSL}(n+1,\mathbb{C})$. Recall that $\sigma:\mathrm{PSU}(2)\to\mathrm{PSU}(n+1)$ is a monomorphism of the Lie group induced by $r_n$(Lemma \ref{ind-embed}).
\begin{lemma}
Denote by $G_v \subset {\rm PSU}(2)$ the monodromy group of a unitary curve $v:S\to\mathbb{P}^1$. Then
\[
\begin{aligned}
M_v&=\{\varphi \in {\rm PSL}(n+1, \mathbb{C}) \mid \varphi\sigma(G_v)\varphi^{-1}\subset\mathrm{PSU}(n+1)\}\\
&=\{\varphi \in {\rm PSL}(n+1, \mathbb{C}) \mid \varphi^*\varphi\in C(\sigma(G_v))\}   
\end{aligned}.
\]
Furthermore, for any unitary curve $v,v_1,v_2:S\to\mathbb{P}^1$, the following properties hold:
\begin{enumerate}
    \item $M_{U\circ v} = \{\sigma(U)\circ\varphi\mid\varphi\in M_v\}$ for $U \in {\rm PSU}(2)$,
    \item $M_{v_1}=M_{v_2}$ if $\overline{G}_{v_1}=\overline{G}_{v_2}$.
\end{enumerate}
\end{lemma}
\begin{proof}
For \( z \in S \), let \( \varrho: \pi_1(S, z) \to \mathrm{PSU}(n+1) \) be the monodromy representation of \( r_n \circ v \). From Section 2, we know that \( \mathrm{Im} \varrho = \sigma(G_v) \). For \( \gamma \in \pi_1(S, z) \), if we extend a branch \( \varphi \circ r_n \circ v_0 \) of \( \varphi \circ r_n \circ v \) along \( \gamma \), we obtain the curve $\varphi \circ \varrho(\gamma) \circ r_n \circ v_0 = (\varphi \varrho(\gamma) \varphi^{-1}) \circ \varphi \circ r_n \circ v_0$. Therefore, the monodromy group is given by \( \varphi \sigma(G_v) \varphi^{-1} \). 

Since we need \( \varphi \circ r_n \circ v \) to be a unitary curve, it follows that \( \varphi U \varphi^{-1} \in \mathrm{PSU}(n+1) \) for all \( U \in \sigma(G_v) \). This implies that $U \varphi^* \varphi = \varphi^* \varphi U$ for all $ U \in \sigma(G_v)$. Thus, \( M_v \) is given by
\[
M_v = \left\{ \varphi \in \mathrm{PSL}(n+1, \mathbb{C}) \mid \varphi^* \varphi \in C(\sigma(G_v)) \right\}.
\]

It is straightforward to verify (1-3).  Moreover, (4) follows from the continuity of the left and right multiplications of 
the Lie group ${\rm PSL}(n+1,\mathbb{C})$.
\end{proof}
\begin{lemma}
    $M_v$ has a decomposition of the form $M_v=\mathrm{PSU}(n+1)\Delta_v$, where $\Delta_v=\left\{ \delta \in \Delta_{n+1} \mid \delta^* \delta \in C(\sigma(G_v)) \right\}$ is a subset of $\Delta_{n+1}$.
\end{lemma}
\begin{proof}
    For \( \varphi \in \mathrm{PSL}(n+1, \mathbb{C}) \), assume \( \varphi =U\circ \delta \), where \( \delta \in \Delta_{n+1} \) and \( U \in \mathrm{PSU}(n+1) \). Then we have \( \varphi^* \varphi = \delta^* \delta \), which implies that \( \varphi \in M_v \) if and only if \( \delta \in M_v \). Therefore, we obtain the decomposition $M_v=\mathrm{PSU}(n+1)\Delta_v$.
\end{proof}
\begin{lemma}[Cholesky factorization]\cite[Corollary 7.2.9]{horn_johnson_2013}\label{lemherm}\quad\\
    The map \( \Delta_{n+1} \to \mathrm{Herm}_{n+1}^+(1), \ \delta \mapsto \delta^* \delta \) is a bijection, where
    \[
    \mathrm{Herm}_{n+1}^+(1)=\left\{H \in \mathrm{SL}(n+1,\mathbb{C})\;\middle|\; H \text{ is positive definite Hermitian} \right\}.
    \]
\end{lemma}
\begin{lemma}[algorithm for Cholesky factorization]\label{lem:Cholesky algorithm}
Let \( M = (m_{i,j})_{0 \leq i,j \leq n} \in \mathrm{Herm}_{n+1}^+(1) \) be a positive semi-definite Hermitian matrix. Define an upper triangular matrix \( \delta = (a_{i,j}) \) (i.e., \( a_{i,j} = 0 \) for \( i > j \)) by:
\begin{enumerate}
    \item For diagonal entries \( j = i \):
  \[
  a_{j,j} = \sqrt{\, m_{j,j} - \sum_{s=0}^{j-1} |a_{s,j}|^2 \,},
  \]
    \item For upper triangular entries \( i < j \):
  \[
  a_{i,j} = \frac{1}{a_{i,i}} \Big( m_{i,j} - \sum_{s=0}^{i-1} \overline{a_{s,i}} \, a_{s,j} \Big).
  \]
\end{enumerate}
Then \( M = \delta^* \delta \), where \( \delta^* \) denotes the conjugate transpose of \( \delta \).
\end{lemma}
This is a well-known classic algorithm for Cholesky factorization, which is easy to check, while I'm not sure what the initial article of it is.
\begin{lemma}\label{lemdivisibility}
    Inherit the notation of the previous lemma. Given $k\in\mathbb{Z}_{>0}$. Then $a_{i,j}=0$ whenever $k\nmid i-j$ if and only if $m_{i,j}=0$ whenever $k\nmid i-j$.
\end{lemma}
\begin{proof}
We prove both directions of the equivalence.

\noindent
\textbf{Direction \((\Rightarrow)\):} 
Assume \(a_{i,j} = 0\) whenever \(k \nmid (i - j)\). Then \(M = \delta^* \delta\) satisfies:
\[
m_{i,j} = \sum_{s=0}^{n} \overline{a_{s,i}}  a_{s,j}.
\]
Fix \(i, j\) such that \(k \nmid (i - j)\). For the term \(\overline{a_{s,i}} a_{s,j}\) to be nonzero, we must have both \(a_{s,i} \neq 0\) and \(a_{s,j} \neq 0\). By the sparsity of \(\delta\), this requires \(k \mid (s - i)\) and \(k \mid (s - j)\). Consequently:
\[
k \mid \left( (s - i) - (s - j) \right) = j - i \implies k \mid (i - j),
\]
contradicting \(k \nmid (i - j)\). Thus \(\overline{a_{s,i}} a_{s,j} = 0\) for all \(s\), so \(m_{i,j} = 0\).

\medskip
\noindent
\textbf{Direction \((\Leftarrow)\):} 
Assume \(m_{i,j} = 0\) whenever \(k \nmid (i - j)\). We prove by induction on \(j\) (from \(0\) to \(n\)) and on \(i\) (from \(0\) to \(j\)) that \(a_{i,j} = 0\) for \(k \nmid (i - j)\).

\begin{itemize}
\item \textit{Base case (\(j = 0\)):} Trivial (no off-diagonal entries).
\item \textit{Inductive step (\(j \geq 1\)):} Assume the claim holds for all columns \(< j\). Then \(a_{0,j}=m_{0,j}=\delta_{j,k}\) holds. For \(j>i\geq 1\), assume the claim holds for all rows \(<i\) when column \(=j\).
  \[
  a_{i,j} = \frac{1}{a_{i,i}} \left( m_{i,j} - \sum_{s=0}^{i-1} \overline{a_{s,i}} a_{s,j} \right).
  \]
  If \(k \nmid (i - j)\), then \(m_{i,j} = 0\). For each \(s \in \{0, \dots, i-1\}\):
  \begin{itemize}
  \item If \(k \nmid (s - i)\), then \(a_{s,i} = 0\) (by induction on column \(i < j\)).
  \item If \(k \nmid (s - j)\), then \(a_{s,j} = 0\) (by induction on row \(s < i\)).
  \item If both nonzero, then \(k \mid (s - i)\) and \(k \mid (s - j)\), implying \(k \mid (i - j)\) (contradiction).
  \end{itemize}
  Thus \(\overline{a_{s,i}} a_{s,j} = 0\), so \(a_{i,j} = 0\). Diagonal entries \(a_{j,j}\) have \(i - j = 0\) (always divisible by \(k\)).
\end{itemize}
\vspace{-0.5em} 
\end{proof}
\begin{lemma}\label{lemdim}
    For any unitary curve \( v:S\to\mathbb{P}^1 \), let us define a subspace 
    $$V_v=\{A \in \mathrm{Mat}_{n+1}(\mathbb{C}) \mid AU = UA, \, \, \forall \, U \in \sigma(G_v) \}$$ of the complex vector space $\mathrm{Mat}_{n+1}(\mathbb{C})$ formed by all $n+1$-order matrices. Then we have $\dim_{\mathbb{R}} \left( \Delta_v \right) = \dim_{\mathbb{C}} \left( V_v \right)-1$.
\end{lemma}
\begin{proof} We divide the proof into the following three steps.
\begin{itemize}

\item Let \( \mathrm{Herm}_{n+1}(1) \) denote the set of \((n+1)\) by $(n+1)$ Hermitian matrices with determinant \(1\). Since the map $\Delta_{n+1} \to \mathrm{Herm}_{n+1}^+(1), \ \delta \mapsto \delta^* \delta$ is a bijection (Lemma \ref{lemherm}), it induces a bijection $\Delta_v \to \mathrm{Herm}_{n+1}^+(1) \cap C(\sigma(G_v))$ by restricting the domain to \( \Delta_v \). Therefore, $\dim_{\mathbb{R}} \left( \Delta_v \right) = \dim_{\mathbb{R}} \left( \mathrm{Herm}_{n+1}^+(1) \cap C(\sigma(G_v)) \right)$. Since \( \mathrm{Herm}_{n+1}^+(1) \) is an open subset of \( \mathrm{Herm}_{n+1}(1) \) and \( \mathrm{Herm}_{n+1}^+(1) \cap C(\sigma(G)) \neq \emptyset \), we conclude that
\[
\dim_{\mathbb{R}} \left( \Delta_v \right) = \dim_{\mathbb{R}} \left( \mathrm{Herm}_{n+1}^+(1) \cap C(\sigma(G_v)) \right) = \dim_{\mathbb{R}} \left( \mathrm{Herm}_{n+1}(1) \cap C(\sigma(G_v)) \right).
\]

\item Let \( \mathrm{Herm}_{n+1} \) denote the set of \((n+1)\) by $(n+1)$ Hermitian matrices, and let \( \mathcal{H}_{n+1} \) be its projection in \( \mathbb{P}(\mathrm{Mat}_{n+1}(\mathbb{C})) \). Since \( \mathrm{Herm}_{n+1}(1) \) is an open dense subset of \( \mathcal{H}_{n+1} \), and \( \mathrm{PSL}(n+1, \mathbb{C}) \) is an open dense subset of \( \mathbb{P}(\mathrm{Mat}_{n+1}(\mathbb{C})) \), it follows that \( \mathrm{Herm}_{n+1}(1) \cap C(\sigma(G_v))=\mathrm{Herm}_{n+1}(1) \cap(\mathrm{PSL}(n+1, \mathbb{C}) \cap\mathbb{P}(V_v))\) is also a non-empty open subset of \( \mathcal{H}_{n+1} \cap \mathbb{P}(V_v) \). Therefore, we conclude that
\[
\dim_{\mathbb{R}} \left( \Delta_v \right) = \dim_{\mathbb{R}} \left( \mathrm{Herm}_{n+1}(1) \cap C(\sigma(G_v)) \right) = \dim_{\mathbb{R}} \left( \mathcal{H}_{n+1} \cap \mathbb{P}(V_v) \right).
\]

\item 
Since \( \mathrm{Mat}_{n+1}(\mathbb{C}) = \mathrm{Herm}_{n+1} \otimes_{\mathbb{R}} \mathbb{C} \), any matrix \( A \in \mathrm{Mat}_{n+1}(\mathbb{C}) \) can be expressed uniquely as \( A = H_1 + \sqrt{-1} H_2 \), where \( H_1, H_2 \in \mathrm{Herm}_{n+1} \). For any  \( U \in \mathrm{SU}(n+1) \), the matrix \( U^* H U \) remains Hermitian for any \( H \in \mathrm{Herm}_{n+1} \). Hence, \( A \in V_v \) if and only if \( H_1, H_2 \in \mathrm{Herm}_{n+1} \cap V_v \). Thus, we can write $V_v = (\mathrm{Herm}_{n+1} \cap V_v) \otimes_{\mathbb{R}} \mathbb{C}$, which implies $\dim_{\mathbb{R}} \left( \mathrm{Herm}_{n+1} \cap V_v \right) = \dim_{\mathbb{C}} \left( V_v \right)$. Since $\mathcal{H}_{n+1} \cap \mathbb{P}(V_v)$ is the projection of $\mathrm{Herm}_{n+1} \cap V_v$ in $\mathbb{P}(\mathrm{Mat}_{n+1}(\mathbb{C}))$, we conclude:
\[
\dim_{\mathbb{R}} \left( \Delta_v \right) = \dim_{\mathbb{R}} \left( \mathcal{H}_{n+1} \cap \mathbb{P}(V_v) \right) = \dim_{\mathbb{R}} \left( \mathrm{Herm}_{n+1} \cap V_v \right)-1 = \dim_{\mathbb{C}} \left( V_v \right) - 1.
\]
\end{itemize}
\end{proof}
\begin{remark}
    Recall that $p:\mathrm{SU}(2)\to\mathrm{PSU}(2)$ is the projection and $V$ is the natural representation of $\mathrm{SU}(2)$. Because $p^{-1}(G_v)$ is a subgroup of $\mathrm{SU}(V)$, we could see $\mathrm{Sym}^n(V)$ as a representation space of $p^{-1}(G_v)$. Then $V_v$ is just the space of $G$-module homomorphisms $\mathrm{End}_{p^{-1}(G_v)}(\mathrm{Sym}^n(V))$.
\end{remark}
Let us recall some results of the representation theory.
\begin{lemma}\label{rep-center}\cite{fulton_harris_2004}
    Let \( G \) be a group. If $V= V_1^{\oplus a_1} \oplus \cdots \oplus V_n^{\oplus a_n}$ is a complex representation of \( G \) , where all \( V_i, i = 1, \dots, n \) are distinct irreducible representation spaces, then $\dim_{\mathbb{C}}\mathrm{End}_G(V)  = a_1^2 + \cdots + a_n^2$.
    In particular, when $G$ is a finite group, $\dim_{\mathbb{C}}\mathrm{End}_G(V)=\frac{1}{|G|}\sum_{g\in G}|\chi_V(g)|^2$, where $\chi_V:G\to\mathbb{C}$ is the character of $V$. 
\end{lemma}

\begin{proof}[Proof of Theorem \ref{thm:fam}]
\quad
\begin{enumerate}
\item $G_v=\mathrm{PSU}(2)$.\\
Notice $\sigma:\mathrm{PSU}(2)=\mathrm{PSU}(V)\to\mathrm{PSU}(n+1)=\mathrm{PSU}(\mathrm{Sym}^n(V))$(Lemma\ref{ind-embed}) is a irreducible representation of $\mathrm{PSU}(2)$. It is irreducible because symmetric product representation $\mathrm{Sym}^n(V)$ is an irreducible representation of $\mathrm{SL}(V)$\cite[Section 11.1]{fulton_harris_2004}. Thus $C(\mathrm{Im}\,\sigma) = \{I\}$ by Schur's Lemma\cite[Lemma 1.7]{fulton_harris_2004}. Therefore, $\delta \in \delta_v$ if and only if $\delta^*\delta = I_{n+1}$, which means $\delta=I_{n+1}$. Consequently, we conclude that $\Delta_v = \{I_{n+1}\}$.
\item $G_v = \mathrm{PU}(1)$.\\
We have $\sigma(G_v) = \left\{\mathrm{diag}(c^n, c^{n-2}, \cdots,c^{-n}) \ \middle|\ |c| = 1\right\}$.
Then $C(\sigma(\mathrm{PU}(1))) = \{\text{all diagonal matrices}\}$. Thus, $\delta \in \Delta_v$ if and only if $\delta^* \delta$ is diagonal. Since $\delta$ is induced by an upper triangular matrix with positive diagonal entries, it must be diagonal with positive entries, which implies
\[
\Delta_v = \{\mathrm{diag}(a_0, \cdots, a_{n}) \in \Delta_{n+1}\}.
\]
It is obviously $\dim_{\mathbb{R}}\Delta_v=n$.
\item $G_v = \mathrm{PO}(2)$\\

Let $g = \begin{pmatrix} 0 & -1 \\ 1 & 0 \end{pmatrix}$. Since $C(\sigma(g))\cap C(\sigma(\mathrm{PU}(1)) = \{ \mathrm{diag}(\lambda_1, \ldots, \lambda_{n+1}) \mid \lambda_i = \lambda_{n+1-i} \}$, it follows that $C(\sigma(G_v)) = \{ \mathrm{diag}(\lambda_1, \ldots, \lambda_{n+1}) \mid \lambda_i = \lambda_{n+1-i} \}$.
Then $\delta\in\Delta_v$ if and only $\delta^*\delta\in C(\sigma(G_v))$. Since $\delta$ is an upper triangular matrix with positive diagonal entries, it implies 
\[
\Delta_v = \{\mathrm{diag}(a_0, \cdots, a_{n})\in \Delta_{n+1} | a_i = a_{n-i}\}.
\]
It is obviously $\dim_{\mathbb{R}}\Delta_v=\left\lfloor \frac{n}{2} \right\rfloor$.
\item $G_v = p(C_k)$.\\
The group $\sigma(G)$ is generated by $\mathrm{diag}(\xi_k^n, \xi_k^{n-2}, \ldots, \xi_k^{-n})$, where $\xi_k$ is a primitive $k$-th root of unity. Then, the centralizer of $\sigma(G_v)$ in ${\rm PSL}(n+1, \mathbb{C})$ is
$C(\sigma(G_v)) = \{ (z_{ij})_{0\leq i,j\leq n} \mid z_{ij} = 0 \text{ if } k \nmid 2(i-j) \}=\{ (z_{ij})_{0\leq i,j\leq n} \mid z_{ij} = 0 \text{ if } \frac{k}{\gcd(k,2)} \nmid (i-j) \}$. Then, $\delta\in \Delta_v$ if and only if $\delta^*\delta\in C(\sigma(G_v))$. Hence, by Lemma \ref{lemdivisibility}
\[
\Delta_v= \{(a_{ij})_{0\leq i,j\leq n} \in \Delta_{n+1} | a_{i,j}= 0 \text{ if } \frac{k}{\gcd(k,2)} \nmid (i-j) \}.
\]
Thus, the number of independent equations given by $C(\sigma(G_v))$ are $|\{(i,j)|0\leq i,j\leq n, \frac{k}{\gcd(k,2)} \nmid (i-j)\}|$. It means $\dim_{\mathbb{R}} \Delta_v=|\{(i,j)|0\leq i,j\leq n, \frac{k}{\gcd(k,2)} \mid (i-j)\}|-1$. The number can be a sum by row:
\[
\begin{aligned}
    \dim_{\mathbb{R}} \Delta_v&=
    2\sum_{i=0}^n\lfloor\frac{i\gcd(k,2)}{k}\rfloor+(n+1)-1\\
    &=-\frac{k}{\gcd(k,2)}\lfloor\frac{n\gcd(k,2)}{k}\rfloor^2+(2n+2-\frac{k}{\gcd(k,2)})\lfloor\frac{n\gcd(k,2)}{k}\rfloor+n
\end{aligned}
\]
\item$G_v = p(D_k)$.\\

In this case, the centralizer $C(\sigma(G_v))$ can be expressed as $C(\sigma(G_v)) = \left\{ (z_{i,j})_{0\leq i,j,\leq n} \in C(\sigma(p(C_{2k}))) \mid z_{i,j} = (-\sqrt{-1})^{i-j}z_{n-i, n-j} \right\}$. Then $\delta\in \Delta_v$ if and only if $\delta^*\delta\in C(\sigma(G_v))$. Hence, by Lemma \ref{lemdivisibility},
\[
\Delta_v = \left\{(a_{i,j})_{0\leq i,j\leq n} \in \Delta_{n+1} \middle| \begin{aligned}
a_{i,j} &= 0 \text{ if } k \nmid i-j\\
    \sum_{l=0}^n \bar{a}_{l,i}a_{l,j} &= (-\sqrt{-1})^{i-j}\sum_{l=0}^n \bar{a}_{l,n-i}a_{l,n-j}
\end{aligned}\right\}.
\]
The number of independent equations given by $C(\sigma(G_v))$ are $|\{(i,j)|0\leq i,j,\leq n,k\nmid i-j\}|+\lfloor|\{(i,j)|0\leq i,j,\leq n,k\mid i-j\}|/2\rfloor$. Thus
\[
\begin{aligned}
    \dim_{\mathbb{R}} \Delta_v &=n^2+2n-|\{(i,j)|0\leq i,j,\leq n,k\nmid i-j\}|-\lfloor|\{(i,j)|0\leq i,j,\leq n,k\mid i-j\}|/2\rfloor\\
    &=\lfloor|\{(i,j)|k\mid i-j\}|/2\rfloor\\
    &=-\frac{k}{2}\lfloor\frac{n}{k}\rfloor^2+(n+1-\frac{k}{2})\lfloor\frac{n}{k}\rfloor+\lfloor\frac{n}{2}\rfloor
\end{aligned}.
\]
\item $G_v=p(E_6)$.\\

Denote $g_1 = \begin{pmatrix}
    0 & -\sqrt{-1} \\
    -\sqrt{-1} & 0
\end{pmatrix},g_2 =\begin{pmatrix}
    \frac{1+\sqrt{-1}}{2} & \frac{1+\sqrt{-1}}{2} \\
    \frac{\sqrt{-1}-1}{2} & \frac{1-\sqrt{-1}}{2}
\end{pmatrix}$.
Let $V$ be the natural representation of $\mathrm{SU}(2)$ and $\rho:E_6\to\mathrm{SU}(V)$ be the given embedding. Then $V_v=\mathrm{End}_{E_6}(\mathrm{Sym}^n(V))$. Since \( E_6 \) is a finite group, direct computation shows that the conjugacy classes of \( E_6 \) are listed in the following table(which is also a well-known result of tetrahedral group):
\begin{table}[h]
\begin{tabular}{|l|l|l|l|l|l|l|l|}
\hline
 conjugacy classes&  $I$ & $-I$  & $g_1$  &  $g_2$&$g_2^2$&$g_2^4$&$g_2^5$ \\ \hline
 their cardinality& 1 & 1 & 6 & 4& 4 & 4 & 4\\ \hline
\end{tabular}
\end{table}\\
Notice that $\chi_{\mathrm{Sym}^n(V)}(g)=\sum_{i=0}^n\lambda_1^i\lambda_2^{n-i}$, where $\lambda_1,\lambda_2$ are two eigenvalues of $\rho(g)$. The character of $\mathrm{Sym}^n(V)$ is as the following table: 
\begin{table}[h]
\begin{tabular}{|l|l|l|l|l|l|l|l|}
\hline
$\mathrm{Sym}^n(V)$& $I$ & $-I$  & $g_1$  &  $g_2$&$g_2^2$&$g_2^4$&$g_2^5$ \\ \hline
$n\equiv1\pmod{2}$ & n+1 & -n-1 & 0 & c & -c & -c & c\\ \hline
$n\equiv0\pmod{2}$ & n+1 & n+1 & $(-1)^{n/2}$ & c & c & c & c\\ \hline
\end{tabular}
\end{table}\\
where $c=\cos{\frac{n\pi}{3}}+\frac{\sqrt{3}}{3}\sin{\frac{n\pi}{3}}$. Thus,
\begin{enumerate}[i.]
    \item When $n$ is odd, by Lemma \ref{rep-center}, we have $\dim_{\mathbb{C}}V_v=\frac{n^2}{12}+\frac{2c^2}{3}+\frac{n}{6}+\frac{1}{12}$. Thus, $\dim_{\mathbb{R}}\Delta_v=\frac{n^2}{12}+\frac{2c^2}{3}+\frac{n}{6}-\frac{11}{12}$.
    \item When $n$ is even, by Lemma \ref{rep-center}, we have $\dim_{\mathbb{C}}V_v=\frac{n^2}{12}+\frac{n}{6}+\frac{2 c^2}{3}+\frac{1}{3}$. Thus, $\dim_{\mathbb{R}}\Delta_v=\frac{n^2}{12}+\frac{n}{6}+\frac{2c^2}{3}-\frac{2}{3}$.
\end{enumerate}
\item $G_v = p(E_7)$. Let 
$
g_1 = \begin{pmatrix}
    0 & \frac{\sqrt{2}(1+\sqrt{-1})}{2} \\
    \frac{\sqrt{2}(\sqrt{-1}-1)}{2} & 0
\end{pmatrix}$
and \\
$g_2 =\begin{pmatrix}
    \frac{1+\sqrt{-1}}{2} & \frac{1+\sqrt{-1}}{2} \\
    \frac{\sqrt{-1}-1}{2} & \frac{1-\sqrt{-1}}{2}
\end{pmatrix}.
$
Let $V$ be the natural representation of $\mathrm{SU}(2)$ and $\rho:E_7\to\mathrm{SU}(V)$ be the given embedding. Then $V_v=\mathrm{End}_{E_7}(\mathrm{Sym}^n(V))$. The conjugacy classes of $E_7$ are listed as(which is also a well-known result of octahedral group)
\begin{table}[h]
\begin{tabular}{|l|l|l|l|l|l|l|l|l|}
\hline
 conjugacy classes&  $I_2$ & $-I_2$  & $(g_1g_2)^2$  &  $g_2$&$g_2^2$&$g_1g_2$&$(g_1g_2)^3$&$g_1$ \\ \hline
 their cardinality& 1 & 1 & 6 & 8& 8 & 6 & 6 &12\\ \hline
\end{tabular}
\end{table}\\
Similar to $E_6$, the character of $\mathrm{Sym}^n(V)$
of $E_7$ are expressed in the following table:
 \begin{table}[h]
\begin{tabular}{|l|l|l|l|l|l|l|l|l|}
\hline
 $\mathrm{Sym}^n(V)$&  $I_2$ & $-I_2$  & $(g_1g_2)^2$  &  $g_2$&$g_2^2$&$g_1g_2$&$(g_1g_2)^3$&$g_1$ \\ \hline
 $n\equiv 1\pmod{2}$ & n+1 & -n-1 & 0 & $c_1$ & $-c_1$ & $-c_2$ & $c_2$ & 0 \\ \hline
 $n\equiv 0\pmod{2}$& n+1 & n+1 & $(-1)^{n/2}$ & $c_1$ & $c_1$ & $c_2$ & $c_2$ & $(-1)^{n/2}$ \\ \hline
\end{tabular}
\end{table}\\
where $c_1=\cos{\frac{n\pi}{3}}+\frac{\sqrt{3}}{3}\sin\,\frac{n\pi}{3}$ and $c_2=\cos{\frac{n\pi}{4}}+\sin{\frac{n\pi}{4}}$. Thus,
\begin{enumerate}[i.]
    \item When $n$ is odd, by Lemma \ref{rep-center}, we have $\dim_{\mathbb{C}}V_v=\frac{n^2}{24}+\frac{n}{12}+\frac{c_1^2}{3}+\frac{c_2^2}{4}+\frac{1}{24}$. Thus, $\dim_{\mathbb{R}}\Delta_v=\frac{n^2}{24}+\frac{n}{12}+\frac{c_1^2}{3}+\frac{c_2^2}{4}-\frac{23}{24}$.
    \item When $n$ is even, by Lemma \ref{rep-center}, we have $\dim_{\mathbb{C}}V_v=\frac{n^2}{24}+\frac{n}{12}+\frac{c_1^2}{3}+\frac{c_2^2}{4}+\frac{5}{12}$. Thus, $\dim_{\mathbb{R}}\Delta_v=\frac{n^2}{24}+\frac{n}{12}+\frac{c_1^2}{3}+\frac{c_2^2}{4}-\frac{7}{12}$.
\end{enumerate}
\item $G_v=p(E_8)$. Let 
\[ 
g_1 = \begin{pmatrix}
    0 & -\sqrt{-1} \\
    -\sqrt{-1} & 0
\end{pmatrix}
\quad \text{and} \quad
g_2 = \begin{pmatrix}
    \frac{1}{2} & -\frac{\sqrt{5}-1}{4} + \frac{\sqrt{5}+1}{4}\sqrt{-1}\\
    \frac{\sqrt{5}-1}{4} + \frac{\sqrt{5}+1}{4}\sqrt{-1} & \frac{1}{2}
\end{pmatrix}.
\]
Let $V$ be the natural representation of $\mathrm{SU}(2)$ and $\rho:E_8\to\mathrm{SU}(V)$ be the given embedding. Then $V_v=\mathrm{End}_{E_8}(\mathrm{Sym}^n(V))$. The conjugacy classes of $E_8$ are listed as(which is also a well-known result of icosahedral group):
\begin{table}[h]
\begin{tabular}{|l|l|l|l|l|l|l|l|l|l|}
\hline
 conjugacy classes&  $I$ & $-I$  & $g_2^2$  &  $g_1$& $(g_1g_2)^2$ & $(g_1g_2)^4$ &$g_2$& $g_1g_2$ & $(g_1g_2)^3$\\ \hline
 their cardinality& 1 & 1 & 20 & 30 & 12 & 12 & 20 & 12 & 12 \\ \hline
\end{tabular}
\end{table}\\
Similar to $E_6$, the character of $\mathrm{Sym}^n(V)$ is expressed in the following table: 
\begin{table}[h]
\begin{tabular}{|l|l|l|l|l|l|l|l|l|l|}
\hline
 $\mathrm{Sym}^n(V)$&  $I$ & $-I$  & $g_2^2$  &  $g_1$ & $(g_1g_2)^2$ & $(g_1g_2)^4$ & $g_2$ & $g_1g_2$ & $(g_1g_2)^3$ \\ \hline
 $n\equiv 1\pmod{1}$&n+1 & -n-1 & $-c_1$ & 0 & $c_3$ & $-c_2$ & $c_1$ & $c_2$ & $-c_3$\\ \hline
 $n\equiv 0\pmod{2}$&n+1 & n+1 & $c_1$ & $(-1)^{n/2}$ & $c_3$ & $c_2$ & $c_1$ & $c_2$ & $c_3$\\ \hline
 \end{tabular}
\end{table}\\
where $c_1=\cos{\frac{n\pi}{3}}+\frac{\sqrt{3}}{3}\sin{\frac{n\pi}{3}}, c_2=\sqrt{1+\frac{2}{\sqrt{5}}} \sin{\frac{n\pi}{5}}+\cos{\frac{n\pi}{5}}$ and $c_3=\sqrt{1-\frac{2}{\sqrt{5}}} \sin{\frac{2n\pi}{5}}+\cos{\frac{2n\pi}{5}}$. Thus,
\begin{enumerate}[i.]
    \item When $n$ is odd, by Lemma \ref{rep-center}, we have $\dim_{\mathbb{C}}V_v=\frac{n^2}{60}+\frac{n}{30}+ \frac{c_1^2}{3}+ \frac{c_2^2}{5}+\frac{c_3^2}{5}+\frac{1}{60}$. Thus, $\dim_{\mathbb{R}}\Delta_v=\frac{n^2}{60}+\frac{n}{30}+ \frac{c_1^2}{3}+ \frac{c_2^2}{5}+\frac{c_3^2}{5}-\frac{59}{60}$.
    \item When $n$ is even, by Lemma \ref{rep-center}, we have $\dim_{\mathbb{C}}V_v=\frac{n^2}{60}+\frac{n}{30}+ \frac{c_1^2}{3}+ \frac{c_2^2}{5}+\frac{c_3^2}{5}+\frac{2}{15}$. Thus, $\dim_{\mathbb{R}}\Delta_v=\frac{n^2}{60}+\frac{n}{30}+ \frac{c_1^2}{3}+ \frac{c_2^2}{5}+\frac{c_3^2}{5}-\frac{13}{15}$.
\end{enumerate}
\end{enumerate}
\end{proof}

\section{Examples}
Let us recall the result of the spherical metric.
\begin{theorem}\cite[Theorem A]{MR3990195}
Let \( g > 0 \) be an integer. Assume \( \beta_1, \dots, \beta_m > 0 \) satisfy
\[
\beta_1 + \dots + \beta_m > 2g - 2 + m,
\]
then there exists a compact orientable Riemann surface \( X \) of genus \( g \) with a spherical metric \(\omega\) on \( X \) that represents \( D = \sum_{j=1}^m (\beta_j - 1) [P_j] \) for some distinct points \( P_1, \dots, P_m \in X \).
\end{theorem}
Then there will be a natural corollary.
\begin{cor}
    Let \(X\), \(\omega\) and \(D\) be the same as above. Then, for each positive integer \( n > 1 \), the SU$(n+1)$ Toda system on \( X \) with cone singularities
$$
\underset{n \, \text{divisors}}{\underbrace{\big(D, D, \ldots, D\big)}}
$$
has a family of reduced solutions, including \( \big(i(n+1-i)\omega\big)_{i=1}^n \) and is characterized in Theorem \ref{thm:fam}.
\end{cor}
\begin{remark}
Consider the $\mathrm{SU}(n+1)$-Toda system with cone singularities
\[
\mathrm{Ric}(\vv{\omega}) = 2 \vv{\omega} C_n + (\delta_{P_1}, \cdots, \delta_{P_m}) \Gamma,
\]
where $\delta_P$ denotes the Dirac measure at $P$ and $\Gamma = (\gamma_{j,i})_{m \times n}$ is a real matrix with $\gamma_{j,i} > -1$. The solution $\vv{\omega}$ represents an $n$-tuple of divisors $(D_i = \sum_{j=1}^m \gamma_{j,i}[P_j])_{i=1}^n$. 
The readers may find the detail of this framework of Toda system with cone singularities in \cite[Section 1]{MSX2024}.
When $\gamma_{j,i} = \beta_j - 1$ for all $i$ and $j$, this corollary shows that the system with cone singularities is solvable.
It should be noted that Lin, Yang and Zhong \cite[Theorem 1.9]{LYZ2020} provide a sufficient condition for the solvability of the Toda system with cone singularities. Our corollary, however, offers a different sufficient condition. These conditions are not equivalent. For example, in the case \( n > 1, \beta_i \in\mathbb{Z}_{>1}, g>0 \), which does not satisfy the condition in \cite[Theorem 1.9]{LYZ2020}, our corollary demonstrates that the system is solvable.
\end{remark}

\noindent\textbf{Acknowledgement:}
Y.S. is supported in part by the National Natural Science Foundation of China (Grant No. 11931009) and the Innovation Program for Quantum Science and Technology (Grant No. 2021ZD0302902).
B.X. would like to express his deep gratitude to  Professor Zhijie Chen at Tsinghua University, Professor Zhaohu Nie at
University of Utah and Professor Guofang Wang at University of Freiburg for their stimulating conversations on Toda systems. Moreover, his research is supported in part by the National Natural Science Foundation of China (Grant Nos. 12271495 and 12071449) and the CAS Project for Young Scientists in Basic Research (YSBR-001).\\


\bibliographystyle{plain}
\bibliography{RefBase}

\end{document}